\documentclass[journal]{IEEEtran}
\usepackage[utf8]{inputenc}
\usepackage{color}
\usepackage{xcolor}
\usepackage{array}
\usepackage{verbatim}
\usepackage{float}
\usepackage{amsmath}
\usepackage{amsthm}
\usepackage{amssymb}
\usepackage{graphicx}
\usepackage{longtable}
\usepackage{multirow}
\usepackage[unicode=true,
bookmarks=false,
breaklinks=false,pdfborder={0 0 1},colorlinks=false]
{hyperref}
\hypersetup{
	colorlinks,bookmarksopen,bookmarksnumbered,citecolor=blue,urlcolor=blue}
\usepackage{cite}

\usepackage{lipsum}
\usepackage{mathtools}
\usepackage{cuted}

\floatstyle{ruled}
\newfloat{algorithm}{tbp}{loa}
\providecommand{\algorithmname}{Algorithm}
\floatname{algorithm}{\protect\algorithmname}

\makeatletter
\let\oldforeign@language\foreign@language
\DeclareRobustCommand{\foreign@language}[1]{%
	\lowercase{\oldforeign@language{#1}}}

\let\oldforeign@language\foreign@language
\DeclareRobustCommand{\foreign@language}[1]{%
	\lowercase{\oldforeign@language{#1}}}

\ifCLASSINFOpdf
\else
\fi

\newcommand{\MYfooter}{\smash{
		\hfil\parbox[t][\height][t]{\textwidth}{\centering
			\thepage}\hfil\hbox{}}}

\def\ps@IEEEtitlepagestyle{%
	\def\@oddhead{\parbox[t][\height][t]{\textwidth}{\centering \scriptsize
			Personal use of this material is permitted. Permission from the author(s) and/or copyright holder(s), must be obtained for all other uses. Please contact us and provide details if you believe this document breaches copyrights.\\
			\noindent\makebox[\linewidth]{}
		}\hfil\hbox{}}%
	\def\@evenhead{\scriptsize\thepage \hfil \leftmark\mbox{}}%
	\def\@oddfoot{\parbox[t][\height][l]{\textwidth}{
			\vspace{-20pt}{\rule{\textwidth}{0.4pt}}\\ \footnotesize\underline{To cite this article:}
			{\bf{\footnotesize\textcolor{red}{H. A. Hashim "GPS-denied Navigation: Attitude, Position, Linear Velocity, and Gravity Estimation with Nonlinear Stochastic Observer," in Proceedings of the 2021 American Control Conference (ACC), New Orleans, LA, USA, 2021, pp. 1146-1151.}}} doi: \href{https://doi.org/10.23919/ACC50511.2021.9482995}{10.23919/ACC50511.2021.9482995}\\
			\noindent\makebox[\linewidth]
		}\hfil\hbox{}}%
	\def\@evenfoot{\MYfooter}}

\makeatother
\newtheorem{defn}{Definition}

\newtheorem{lem}{Lemma}

\newtheorem{thm}{Theorem}

\newtheorem{assum}{Assumption}

\begin{document}
	\bstctlcite{IEEEexample:BSTcontrol}

	\title{GPS-denied Navigation: Attitude, Position, Linear Velocity, and Gravity Estimation with Nonlinear Stochastic Observer}

\author{Hashim A. Hashim% <-this % stops a space
	\thanks{This work was supported in part by Thompson Rivers University Internal research fund \# 102315.}
	\thanks{Corresponding author, H. A. Hashim is with the Department of Engineering and Applied Science, Thompson Rivers University, Kamloops, British Columbia, Canada, V2C-0C8, e-mail: hhashim@tru.ca}
}

% \markboth{IEEE TRANSACTIONS ON INTELLIGENT TRANSPORTATION SYSTEMS, \today}{Hashim \MakeLowercase{\textit{et al.}}: Landmark and IMU Data Fusion: Systematic Convergence Geometric Nonlinear Observer for SLAM and Velocity Bias}

% \markboth{}{Hashim \MakeLowercase{\textit{et al.}}: Nonlinear Filter for Simultaneous Localization and Mapping on a Matrix Lie Group using IMU and Feature Measurements}

\maketitle

\begin{abstract}
Successful navigation of a rigid-body traveling with six degrees of
freedom (6 DoF) requires accurate estimation of attitude, position,
and linear velocity. The true navigation dynamics are highly nonlinear
and are modeled on the matrix Lie group of $\mathbb{SE}_{2}(3)$.
This paper presents novel geometric nonlinear continuous stochastic
navigation observers on $\mathbb{SE}_{2}(3)$ capturing the true nonlinearity
of the problem. The proposed observers combines IMU and landmark measurements.
It efficiently handles the IMU measurement noise. The proposed observers
are guaranteed to be almost semi-globally uniformly ultimately bounded
in the mean square. Quaternion representation is provided. A real-world
quadrotor measurement dataset is used to validate the effectiveness
of the proposed observers in its discrete form.
\end{abstract}

% Note that keywords are not normally used for peerreview papers.

\IEEEpeerreviewmaketitle{}

\section{Introduction}

\IEEEPARstart{A}{utonomous} navigation would be infeasible without robust algorithms
that enable accurate pose (\textit{i.e.} attitude and position) and
velocity estimation of a rigid-body. The challenge of attitude estimation,
an essential component of pose, has been explored extensively over
the past three decades \cite{hashim2019SO3Wiley,crassidis2003unscented,grip2012attitude,hashim2018SO3Stochastic,lee2012exponential}.
Attitude can be defined given at least two observations in the inertial-frame
and their measurements in the body-frame. For instance, a typical
low-cost inertial measurement unit (IMU) module includes a magnetometer
and an accelerometer which supply the two necessary body-frame measurements
and a gyroscope which provides measurements of angular velocity. The
main shortcoming of low-cost sensors is the high level of noise. Multiple
solutions tackle attitude measurement uncertainty in attitude estimation,
namely Gaussian filters \cite{crassidis2003unscented}, nonlinear
deterministic filters on the Special Orthogonal Group $\mathbb{SO}(3)$
\cite{grip2012attitude,lee2012exponential}, and nonlinear stochastic
filters on $\mathbb{SO}(3)$ \cite{hashim2018SO3Stochastic,hashim2019SO3Wiley}.
In contrast, pose estimation requires a vision unit in addition to
an IMU. Initially, Gaussian filters dominated the area of pose estimation.
In the last few years, nonlinear deterministic filters on the Special
Euclidean Group $\mathbb{SE}(3)$ \cite{hashim2019SE3Det,vasconcelos2010nonlinear}
and nonlinear stochastic filters on $\mathbb{SE}(3)$ \cite{hashim2020SE3Stochastic,hashim2020LetterSLAM,hashim2020TITS_SLAM}
have been shown to be more effective. Pose estimation algorithms rely
on measurements of angular and linear velocity \cite{hashim2020SE3Stochastic}.
In practice, linear velocity information is not attainable in case
of 1) a GPS-denied environment and 2) a vehicle equipped with low-cost
sensors.

The true six degrees of freedom (6 DoF) navigation dynamics are a
combination of attitude, position, and linear velocity dynamics. The
dynamics are highly nonlinear, are modeled on the Lie group of $\mathbb{SE}_{2}(3)$,
are neither right nor left invariant, and rely on angular velocity
and acceleration. The navigation problem has been approached using
Gaussian filters, for instance, \cite{batista2011observability}.
Other solutions that attempted mimicking the true navigation dynamics
include a Riccati observer \cite{hua2018riccati} and an invariant
extended Kalman filter (IEKF) on $\mathbb{SE}_{2}(3)$ \cite{barrau2016invariant}. 

Considering the true nature of the navigation dynamics, this paper
proposes novel nonlinear stochastic navigation observers on $\mathbb{SE}_{2}(3)$
that 1) mimics the true navigation dynamics; 2) estimates rigid-body's
attitude, position, linear velocity, and unknown gravity; and 3) relies
on measurements of angular velocity and acceleration. The noise associated
with IMU measurements is successfully handled. The closed loop signals
are guaranteed to be almost semi-globally uniformly ultimately bounded
(SGUUB) in the mean square. The novel solutions is shown to be cost-effective
at a low sampling rate using a real-world dataset. 

The rest of the paper is organized as follows: Section \ref{sec:SE3_Problem-Formulation}
introduces important math notation and the preliminaries, and formulates
the navigation problem in a stochastic sense. Section \ref{sec:Non-Nav-Observer1}
presents a novel nonlinear stochastic navigation observer. Section
\ref{sec:SE3_Simulations} presents the obtained results. Finally,
Section \ref{sec:SE3_Conclusion} summarizes the work.

\section{Problem Formulation\label{sec:SE3_Problem-Formulation}}

\subsection{Preliminaries}

In this paper, sets of real numbers, nonnegative real numbers, and
a real $n$-by-$m$ dimensional space are defined by $\mathbb{R}$,
$\mathbb{R}_{+}$, and $\mathbb{R}^{n\times m}$, respectively. For
$x\in\mathbb{R}^{n}$ and $M\in\mathbb{R}^{n\times m}$, $||x||=\sqrt{x^{\top}x}$
is the Euclidean norm of $x$ and $||M||_{F}=\sqrt{{\rm Tr}\{MM^{*}\}}$
is the Frobenius norm of $M$ with $*$ being a conjugate transpose.
$\mathbf{I}_{n}$ denotes an $n$-by-$n$ identity matrix, while $0_{n\times m}$
and $1_{n\times m}$ denote $n$-by-$m$ dimensional matrices of zeros
and ones, respectively. A set of eigenvalues of $M\in\mathbb{R}^{n\times n}$
is described by $\lambda(M)=\{\lambda_{1},\lambda_{2},\ldots,\lambda_{n}\}$
with $\overline{\lambda}_{M}=\overline{\lambda}(M)$ being the maximum
value and $\underline{\lambda}_{M}=\underline{\lambda}(M)$ being
the minimum value of $\lambda(M)$. $\mathbb{P}\{\cdot\}$ represents
probability and $\mathbb{E}[\cdot]$ denotes an expected value. $\left\{ \mathcal{I}\right\} $
represents fixed inertial-frame and $\left\{ \mathcal{B}\right\} $
denotes fixed body-frame. Rigid-body's attitude is defined by $R\in\mathbb{SO}\left(3\right)$
where $\mathbb{SO}(3)=\{R\in\mathbb{R}^{3\times3}|RR^{\top}=R^{\top}R=\mathbf{I}_{3}\text{, }{\rm det}(R)=+1\}$
with ${\rm det}(\cdot)$ being a determinant. $\mathfrak{so}(3)$
defines the Lie algebra of $\mathbb{SO}(3)$ with $\mathfrak{so}\left(3\right)=\{[x]_{\times}\in\mathbb{R}^{3\times3}|[x]_{\times}^{\top}=-[x]_{\times},x\in\mathbb{R}^{3}\}$.
Note that $[x]_{\times}$ is a skew symmetric matrix such tha{\small{}t
	\begin{align*}
	\left[x\right]_{\times} & =\left[\begin{array}{ccc}
	0 & -x_{3} & x_{2}\\
	x_{3} & 0 & -x_{1}\\
	-x_{2} & x_{1} & 0
	\end{array}\right]\in\mathfrak{so}\left(3\right),\hspace{1em}x=\left[\begin{array}{c}
	x_{1}\\
	x_{2}\\
	x_{3}
	\end{array}\right]
	\end{align*}
}The inverse mapping of $[\cdot]_{\times}$ follows the map $\mathbf{vex}:\mathfrak{so}\left(3\right)\rightarrow\mathbb{R}^{3}$
with $\mathbf{vex}([x]_{\times})=x,\forall x\in\mathbb{R}^{3}$. The
anti-symmetric projection on $\mathfrak{so}\left(3\right)$ is represented
by $\boldsymbol{\mathcal{P}}_{a}(A)=\frac{1}{2}(A-A^{\top})\in\mathfrak{so}\left(3\right),\forall A\in\mathbb{R}^{3\times3}$.
$\boldsymbol{\Upsilon}=\mathbf{vex}\circ\boldsymbol{\mathcal{P}}_{a}$
describes a composition mapping where $\boldsymbol{\Upsilon}(A)=\mathbf{vex}(\boldsymbol{\mathcal{P}}_{a}(A))\in\mathbb{R}^{3},\forall A\in\mathbb{R}^{3\times3}$.
The Euclidean distance of $R\in\mathbb{SO}\left(3\right)$ is described
by
\begin{equation}
||R||_{{\rm I}}={\rm Tr}\{\mathbf{I}_{3}-R\}/4\in\left[0,1\right]\label{eq:NAV_Ecul_Dist}
\end{equation}
where $-1\leq{\rm Tr}\{R\}\leq3$ and $||R||_{{\rm I}}=\frac{1}{8}||\mathbf{I}_{3}-R||_{F}^{2}$,
see \cite{hashim2018SO3Stochastic,hashim2019SO3Wiley}. Consider a
rigid-body navigating in 3D space where its attitude, position, and
velocity are termed $R\in\mathbb{SO}\left(3\right)$, $P\in\mathbb{R}^{3}$,
and $V\in\mathbb{R}^{3}$, respectively, with $R\in\{\mathcal{B}\}$
and $P,V\in\{\mathcal{I}\}$. Let $\mathbb{SE}_{2}\left(3\right)=\mathbb{SO}\left(3\right)\times\mathbb{R}^{3}\times\mathbb{R}^{3}\subset\mathbb{R}^{5\times5}$
be the extended form of $\mathbb{SE}\left(3\right)=\mathbb{SO}\left(3\right)\times\mathbb{R}^{3}\subset\mathbb{R}^{4\times4}$
with
\begin{align}
\mathbb{SE}_{2}\left(3\right) & =\{\left.X\in\mathbb{R}^{5\times5}\right|R\in\mathbb{SO}\left(3\right),P,V\in\mathbb{R}^{3}\}\label{eq:NAV_SE2_3}\\
X=\Psi( & R,P,V)=\left[\begin{array}{ccc}
R & P & V\\
0_{1\times3} & 1 & 0\\
0_{1\times3} & 0 & 1
\end{array}\right]\in\mathbb{SE}_{2}\left(3\right)\label{eq:NAV_X}
\end{align}
$X\in\mathbb{SE}_{2}\left(3\right)$ denotes a homogeneous navigation
matrix which is assumed to be bounded. The tangent space of $\mathbb{SE}_{2}\left(3\right)$
at point $X$ is $T_{X}\mathbb{SE}_{2}\left(3\right)\in\mathbb{R}^{5\times5}$.
Define a submanifold $\mathcal{U}_{\mathcal{M}}=\mathfrak{so}\left(3\right)\times\mathbb{R}^{3}\times\mathbb{R}^{3}\times\mathbb{R}\subset\mathbb{R}^{5\times5}$
such that
\begin{align}
\mathcal{U}_{\mathcal{M}} & =\left\{ \left.u([\Omega\text{\ensuremath{]_{\times}}},V,a,\kappa)\right|[\Omega\text{\ensuremath{]_{\times}}}\in\mathfrak{so}(3),V,a\in\mathbb{R}^{3},\kappa\in\mathbb{R}\right\} \nonumber \\
u( & [\Omega\text{\ensuremath{]_{\times}}},V,a,\kappa)=\left[\begin{array}{ccc}
[\Omega\text{\ensuremath{]_{\times}}} & V & a\\
0_{1\times3} & 0 & 0\\
0_{1\times3} & \kappa & 0
\end{array}\right]\in\mathcal{U}_{\mathcal{M}}\subset\mathbb{R}^{5\times5}\label{eq:NAV_u}
\end{align}
with $\Omega\in\mathbb{R}^{3}$, $V\in\mathbb{R}^{3}$, and $a\in\mathbb{R}^{3}$
being the rigid-body's true angular velocity, linear velocity, and
apparent acceleration composed of all non-gravitational forces affecting
the rigid-body, respectively, where $\Omega,a\in\{\mathcal{B}\}$.

\subsection{Measurements and Dynamics}

\begin{figure}[h]
	\vspace{-0.1cm}
	\centering{}\includegraphics[scale=0.38]{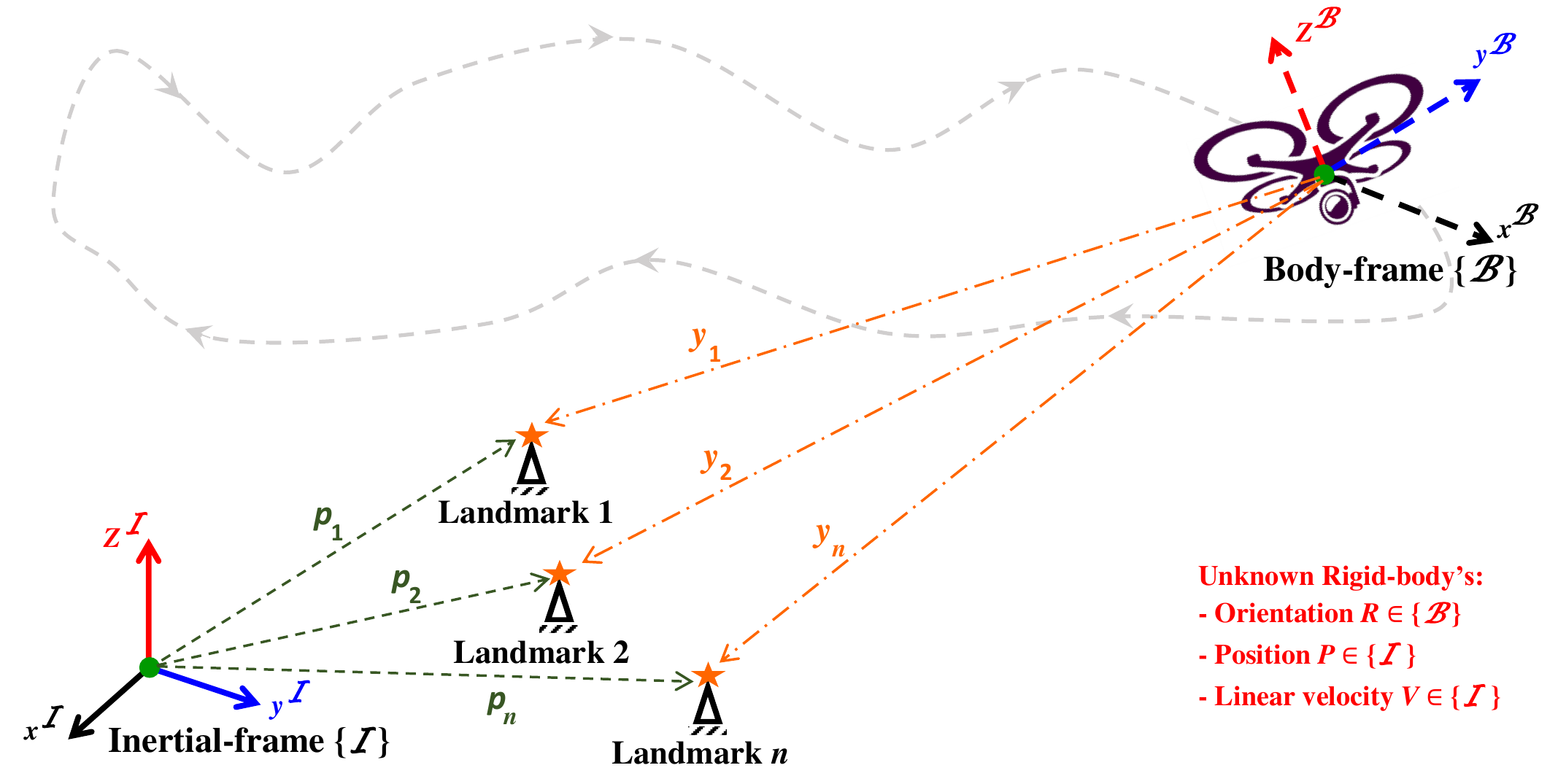}\vspace{-0.1cm}
	\caption{Navigation estimation problem.}
	\label{fig:NAVIGATION}\vspace{-0.1cm}
\end{figure}
The true dynamics of the homogeneous navigation matrix in \eqref{eq:NAV_X}
are as follow:\vspace{-0.1cm}
\begin{equation}
\begin{cases}
\dot{R} & =R\left[\Omega\right]_{\times}\\
\dot{P} & =V\\
\dot{V} & =Ra+\overrightarrow{\mathtt{g}}
\end{cases}\hspace{1em}\equiv\hspace{1em}\underbrace{\dot{X}=XU-\mathcal{\mathcal{G}}X}_{\text{Compact form}}\label{eq:NAV_Detailed_True_dot}
\end{equation}
with $\overrightarrow{\mathtt{g}}$ denoting a gravity vector. The
left portion of \eqref{eq:NAV_Detailed_True_dot} represents the detailed
navigation dynamics, while the right portion is its equivalent compact
form with $X=\left[\begin{array}{ccc}
R & P & V\\
0_{1\times3} & 1 & 0\\
0_{1\times3} & 0 & 1
\end{array}\right]\in\mathbb{SE}_{2}\left(3\right)$, $U=u([\Omega\text{\ensuremath{]_{\times}}},0_{3\times1},a,1)=\left[\begin{array}{ccc}
[\Omega\text{\ensuremath{]_{\times}}} & 0_{3\times1} & a\\
0_{1\times3} & 0 & 0\\
0_{1\times3} & 1 & 0
\end{array}\right]\in\mathcal{U}_{\mathcal{M}}$, and $\mathcal{\mathcal{G}}=u(0_{3\times3},0_{3\times1},-\overrightarrow{\mathtt{g}},1)=\left[\begin{array}{ccc}
0_{3\times3} & 0_{3\times1} & -\overrightarrow{\mathtt{g}}\\
0_{1\times3} & 0 & 0\\
0_{1\times3} & 1 & 0
\end{array}\right]\in\mathcal{U}_{\mathcal{M}}$, visit \eqref{eq:NAV_u}. Note that $\dot{X}:\mathbb{SE}_{2}\left(3\right)\times\mathcal{U}_{\mathcal{M}}\rightarrow T_{X}\mathbb{SE}_{2}\left(3\right)$.
The components of the navigation matrix $X$, namely $R$, $P$, and
$V$, become completely unknown when the vehicle is equipped with
low-cost sensors in a GPS-denied environment. Fig. \ref{fig:NAVIGATION}
depicts the navigation problem. Availability of a set of sensor measurements
enables the estimation of $X$. Consider a group of $n$ landmarks
observed in $\left\{ \mathcal{I}\right\} $ and measured in $\left\{ \mathcal{B}\right\} $
\cite{hashim2020SE3Stochastic,hashim2020LetterSLAM,hashim2020TITS_SLAM}:%
\begin{comment}
such that
\end{comment}
\vspace{-0.1cm}
\begin{align}
\overline{y}_{i} & =X^{-1}\overline{p}_{i}+[(n_{i}^{y})^{\top},0,0]^{\top}\in\mathbb{R}^{5}\nonumber \\
y_{i} & =R^{\top}(p_{i}-P)+n_{i}^{y}\in\mathbb{R}^{3}\label{eq:NAV_Vec_Landmark}
\end{align}
\vspace{-0.1cm}
where $X^{-1}=\Psi(R^{\top},-R^{\top}P,-R^{\top}V)$, $p_{i}\in\left\{ \mathcal{I}\right\} $
denotes the $i$th landmark observation, $y_{i}\in\left\{ \mathcal{B}\right\} $
denotes the $i$th landmark measurement, and $n_{i}^{y}\in\left\{ \mathcal{B}\right\} $
denotes the noise associated with $y_{i}$, $\overline{y}_{i}=[y_{i}^{\top},1,0]^{\top}$,
and $\overline{p}_{i}=[p_{i}^{\top},1,0]^{\top}$. Note that in this
analysis, $n_{i}^{y}$ is assumed to be zero.

\begin{assum}\label{Assum:NAV_1Landmark}The number of non-collinear
	landmarks available for observation and measurement is greater or
	equal to three.\end{assum}

Measurements of $\Omega$ and $a$ are easily obtainable by a low-cost
IMU module:\vspace{-0.1cm}
\begin{equation}
\begin{cases}
\Omega_{m} & =\Omega+n_{\Omega}\in\mathbb{R}^{3}\\
a_{m} & =a+n_{a}\in\mathbb{R}^{3}
\end{cases}\label{eq:NAV_XVelcoity}
\end{equation}
with $n_{\Omega}$ and $n_{a}$ being unknown bounded zero-mean noise.
A derivative of a Gaussian process is a Gaussian process. Therefore,
one can express $n_{\Omega}=\mathcal{Q}d\beta_{\Omega}/dt$ and $n_{a}=\mathcal{Q}d\beta_{a}/dt$
as Brownian motion process vectors \cite{deng2001stabilization,ito1984lectures}
with $\mathcal{Q}={\rm diag}(\mathcal{Q}_{1,1},\mathcal{Q}_{2,2},\mathcal{Q}_{3,3})\in\mathbb{R}^{3\times3}$
being an unknown positive time-variant diagonal matrix. $\mathcal{Q}^{2}=\mathcal{Q}\mathcal{Q}^{\top}$
denotes the covariance of $n_{\Omega}$ and $n_{a}$. For more details
on the Brownian motion properties with regard to the attitude and
pose estimation problems visit \cite{hashim2018SO3Stochastic,hashim2019SO3Wiley,hashim2020SE3Stochastic}.
Hence, the attitude dynamics in \eqref{eq:NAV_Detailed_True_dot}
can be represented in an incremental form as $dR=R[\Omega_{m}]_{\times}dt-R[\mathcal{Q}d\beta_{\Omega}]_{\times}$.
From \eqref{eq:NAV_Ecul_Dist}, the dynamics in \eqref{eq:NAV_Detailed_True_dot}
may be rewritten as a stochastic differential equation:\vspace{-0.1cm}
\begin{equation}
\begin{cases}
d||R||_{{\rm I}} & =(1/2)\mathbf{vex}(\boldsymbol{\mathcal{P}}_{a}(R))^{\top}(\Omega_{m}dt-\mathcal{Q}d\beta_{\Omega})\\
dP & =Vdt\\
dV & =(Ra_{m}+\overrightarrow{\mathtt{g}})dt-R\mathcal{Q}d\beta_{a}
\end{cases}\label{eq:NAV_Detailed_True_STCH_dot}
\end{equation}
where ${\rm Tr}\{R[\Omega_{m}]_{\times}\}={\rm Tr}\{\boldsymbol{\mathcal{P}}_{a}(R)[\Omega_{m}]_{\times}\}=-2\mathbf{vex}(\boldsymbol{\mathcal{P}}_{a}(R))^{\top}\Omega_{m}$,
visit \cite{hashim2019SO3Wiley,hashim2020SE3Stochastic,hashim2021AESCTE}.
In other words, \eqref{eq:NAV_Detailed_True_STCH_dot} can be described
as\vspace{-0.1cm}
\begin{align}
dx & =fdt+h\overline{\mathcal{Q}}d\beta\label{eq:NAV_Detailed_True_STCH_dot-1}
\end{align}
where $x=[||R||_{{\rm I}},P^{\top},V^{\top}]^{\top}\in\mathbb{R}^{7}$,
$f=[(1/2)\mathbf{vex}(\boldsymbol{\mathcal{P}}_{a}(R))^{\top}\Omega_{m},V^{\top},(Ra_{m}+\overrightarrow{\mathtt{g}})^{\top}]^{\top}\in\mathbb{R}^{7}$,
$h\in\mathbb{R}^{7\times9}$, $\overline{\mathcal{Q}}d\beta=[d\beta_{\Omega}^{\top}\mathcal{Q},0_{3\times1}^{\top},d\beta_{a}^{\top}\mathcal{Q}]^{\top}\in\mathbb{R}^{9}$,
$\overline{\mathcal{Q}}={\rm diag}(\mathcal{Q},\mathcal{Q},\mathcal{Q})\in\mathbb{R}^{9\times9}$,
and $\beta=[\beta_{\Omega}^{\top},0_{3\times1}^{\top},\beta_{a}^{\top}]^{\top}\in\mathbb{R}^{9}$.
Note that ${\rm diag}$ denotes a diagonal matrix. With the aim of
achieving adaptive stabilization, define\vspace{-0.1cm}
\begin{equation}
\sigma=[\sup_{t\geq0}\mathcal{Q}_{1,1},\sup_{t\geq0}\mathcal{Q}_{2,2},\sup_{t\geq0}\mathcal{Q}_{3,3}]^{\top}\in\mathbb{R}^{3}\label{eq:NAV_s}
\end{equation}

\begin{defn}
	\label{def:NAV_SGUUB}\cite{hashim2018SO3Stochastic,hashim2020SE3Stochastic,ji2006adaptive,hashim2021AESCTE}
	For the stochastic dynamics in \eqref{eq:NAV_Detailed_True_STCH_dot-1},
	$x(t)$ is almost SGUUB if for a known set $\varDelta\in\mathbb{R}^{7}$
	and $x(t_{0})$ there is a constant $c>0$ and a time constant $\tau_{c}=\tau_{c}(\kappa,x(t_{0}))$
	with $\mathbb{E}[||x(t_{0})||]<c,\forall t>t_{0}+c$.%
	\begin{comment}
	%
	\begin{defn}
	\label{def:NAV_LV}For the stochastic dynamics in \eqref{eq:NAV_Detailed_True_STCH_dot-1},
	let $\mathbb{V}(x)$ be a twice differentiable function $\mathbb{V}(x)\in\mathcal{C}^{2}$.
	The differential operator $\mathcal{L}\mathbb{V}(x)$ is defined as
	follows: 
	\[
	\mathcal{L}\mathbb{V}(x)=\mathbb{V}_{x}^{\top}f+\frac{1}{2}{\rm Tr}\{h\overline{\mathcal{Q}}^{2}h^{\top}\mathbb{V}_{xx}\}
	\]
	where $\mathbb{V}_{x}=\partial\mathbb{V}/\partial x$ and $\mathbb{V}_{xx}=\partial^{2}\mathbb{V}/\partial x^{2}$.
	\end{defn}
	\end{comment}
\end{defn}
\begin{lem}
	\label{Lemm:NAV_deng}\cite{deng2001stabilization} Consider the stochastic
	system in \eqref{eq:NAV_Detailed_True_STCH_dot-1} and suppose that
	$\mathbb{V}(x)$ is a twice differentiable cost function with $\mathbb{V}:\mathbb{R}^{7}\rightarrow\mathbb{R}_{+}$
	such that
	\begin{equation}
	\mathcal{L}\mathbb{V}(x)=\mathbb{V}_{x}^{\top}f+\frac{1}{2}{\rm Tr}\{h\overline{\mathcal{Q}}^{2}h^{\top}\mathbb{V}_{xx}\}\label{eq:NAV_Vfunction_Lyap0}
	\end{equation}
	where $\mathcal{L}\mathbb{V}(x)$ denotes a differential operator,
	$\mathbb{V}_{x}=\partial\mathbb{V}/\partial x$ and $\mathbb{V}_{xx}=\partial^{2}\mathbb{V}/\partial x^{2}$.
	Define $\varpi_{1}(\cdot)$ and $\varpi_{2}(\cdot)$ as class $\mathcal{K}_{\infty}$
	functions and let constants $\eta_{1}>0$ and $\eta_{2}\geq0$ such
	that
	\begin{align}
	& \hspace{1em}\varpi_{1}(x)\leq\mathbb{V}(x)\leq\varpi_{2}(x)\label{eq:NAV_Vfunction_Lyap}\\
	\mathcal{L}\mathbb{V}(x)= & \mathbb{V}_{x}^{\top}f+\frac{1}{2}{\rm Tr}\{h\overline{\mathcal{Q}}^{2}h^{\top}\mathbb{V}_{xx}\}\leq-\eta_{1}\mathbb{V}(x)+\eta_{2}\label{eq:NAV_dVfunction_Lyap}
	\end{align}
	Thus, the stochastic differential system in \eqref{eq:NAV_Detailed_True_STCH_dot}
	has almost a unique strong solution on $[0,\infty)$. Additionally,
	the solution $x$ is bounded in probability satisfying
	\begin{equation}
	\mathbb{E}[\mathbb{V}(x)]\leq\mathbb{V}(x(0)){\rm exp}(-\eta_{1}t)+\eta_{2}/\eta_{1}\label{eq:NAV_EVfunction_Lyap}
	\end{equation}
	Furthermore, the inequality in \eqref{eq:NAV_EVfunction_Lyap} shows
	that $x$ is SGUUB in the mean square. 
\end{lem}

\subsection{Estimates, Error, and Measurements Setup\label{subsec:Navigation-Matrix}}

Consider $\hat{\sigma}$ to be the estimate of $\sigma$ described
in \eqref{eq:NAV_s}. Let the covariance error be
\[
\tilde{\sigma}=\sigma-\hat{\sigma}\in\mathbb{R}^{3}
\]
Let the estimate of $X\in\mathbb{SE}_{2}\left(3\right)$ in \eqref{eq:NAV_X}
be
\begin{equation}
\hat{X}=\Psi(\hat{R},\hat{P},\hat{V})=\left[\begin{array}{ccc}
\hat{R} & \hat{P} & \hat{V}\\
0_{1\times3} & 1 & 0\\
0_{1\times3} & 0 & 1
\end{array}\right]\in\mathbb{SE}_{2}\left(3\right)\label{eq:NAV_X_est}
\end{equation}
where $\hat{R}\in\mathbb{SO}\left(3\right)$, $\hat{P}\in\mathbb{R}^{3}$,
and $\hat{V}\in\mathbb{R}^{3}$ refer to the estimates of $R$, $P$,
and $V$, respectively. Define the error between $X$ and $\hat{X}$
as
\begin{align*}
\tilde{X}=X\hat{X}^{-1} & =\left[\begin{array}{ccc}
\tilde{R} & \tilde{P} & \tilde{V}\\
0_{1\times3} & 1 & 0\\
0_{1\times3} & 0 & 1
\end{array}\right]\in\mathbb{SE}_{2}\left(3\right)
\end{align*}
such that $\hat{X}^{-1}=\Psi(\hat{R}^{\top},-\hat{R}^{\top}\hat{P},-\hat{R}^{\top}\hat{V})$,
$\tilde{R}=R\hat{R}^{\top}$, $\tilde{P}=P-\tilde{R}\hat{P}$, and
$\tilde{V}=V-\tilde{R}\hat{V}$. The overarching objective of driving
$X\rightarrow\hat{X}$ means that $\tilde{X}\rightarrow\mathbf{I}_{5}$,
$\tilde{R}\rightarrow\mathbf{I}_{3}$, $\tilde{P}\rightarrow0_{3\times1}$,
and $\tilde{V}\rightarrow0_{3\times1}$. Let us define the error as
\begin{align*}
\overset{\circ}{\tilde{y}}_{i} & =\overline{p}_{i}-\tilde{X}^{-1}\overline{p}_{i}=\overline{p}_{i}-\hat{X}\overline{y}_{i}=[(p_{i}-\hat{R}y_{i}-\hat{P})^{\top},0,0]^{\top}
\end{align*}
where $\overset{\circ}{\tilde{y}}_{i}=[\tilde{y}_{i}^{\top},0,0]^{\top}$,
$p_{i}-\hat{R}y_{i}-\hat{P}=\tilde{p}_{i}-\tilde{P}$, $\tilde{p}_{i}=\hat{p}_{i}-\tilde{R}p_{i}$,
and $\tilde{P}=\hat{P}-\tilde{R}P$. Let $s_{i}>0$ be the sensor
confidence level of the $i$th measurement. Define the following elements
in the context of available vector measurements:{\small{}
	\begin{equation}
	\begin{cases}
	p_{c} & =\frac{1}{s_{T}}\sum_{i=1}^{n}s_{i}p_{i},\hspace{1em}s_{T}=\sum_{i=1}^{n}s_{i}\\
	M & =\sum_{i=1}^{n}s_{i}(p_{i}-p_{c})(p_{i}-p_{c})^{\top}\\
	& =\sum_{i=1}^{n}s_{i}p_{i}p_{i}^{\top}-s_{T}p_{c}p_{c}^{\top}\\
	\tilde{R}^{\top}\tilde{P}_{\varepsilon} & =\sum_{i=1}^{n}s_{i}\tilde{y}_{i}=\frac{1}{s_{T}}\sum_{i=1}^{n}s_{i}(p_{i}-\hat{R}y_{i}-\hat{P})\\
	M\tilde{R} & =\sum_{i=1}^{n}s_{i}(p_{i}-p_{c})(p_{i}-P)^{\top}\tilde{R}\\
	& =\sum_{i=1}^{n}s_{i}(p_{i}-p_{c})y_{i}^{\top}\hat{R}^{\top}
	\end{cases}\label{eq:NAV_Set_Measurements}
	\end{equation}
}It can be deduced that $\tilde{R}\rightarrow\mathbf{I}_{3}$ indicates
that $\tilde{P}_{\varepsilon}\rightarrow\tilde{P}$.

\begin{lem}
	\label{Lemm:SLAM_Lemma1}Let $\tilde{R}\in\mathbb{SO}\left(3\right)$
	and $M=M^{\top}\in\mathbb{R}^{3\times3}$ as in \eqref{eq:NAV_Set_Measurements}.
	Consider $\overline{\mathbf{M}}={\rm Tr}\{M\}\mathbf{I}_{3}-M$ where
	$\underline{\lambda}_{\overline{\mathbf{M}}}$ and $\overline{\lambda}_{\overline{\mathbf{M}}}$
	denote the minimum and the maximum eigenvalues of $\overline{\mathbf{M}}$,
	respectively. Since $||M\tilde{R}||_{{\rm I}}=\frac{1}{4}{\rm Tr}\{M(\mathbf{I}_{3}-\tilde{R})\}$
	and $\boldsymbol{\Upsilon}(M\tilde{R})=\mathbf{vex}(\boldsymbol{\mathcal{P}}_{a}(M\tilde{R}))$,
	the{\small{}n
		\begin{align}
		\frac{\underline{\lambda}_{\overline{\mathbf{M}}}}{2}(1+{\rm Tr}\{\tilde{R}\})||M\tilde{R}||_{{\rm I}} & \leq||\boldsymbol{\Upsilon}(M\tilde{R})||^{2}\leq2\overline{\lambda}_{\overline{\mathbf{M}}}||M\tilde{R}||_{{\rm I}}\label{eq:SLAM_lemm1_1}
		\end{align}
	}\begin{proof}See (\cite{hashim2019SE3Det}, Lemma 1).\end{proof}
\end{lem}
In view of Lemma \ref{Lemm:SLAM_Lemma1}, let $\lambda(M)=\{\lambda_{1},\lambda_{2},\lambda_{3}\}$
with $\lambda_{3}\geq\lambda_{2}\geq\lambda_{1}$. According to Assumption
\ref{Assum:NAV_1Landmark}, all the eigenvalues of $\lambda(M)$ are
nonnegative, and at least $\lambda_{2}$ and $\lambda_{3}$ are greater
than zero. As such, (\cite{bullo2004geometric} page. 553): 1) $\overline{\mathbf{M}}$
is positive-definite, and 2) $\lambda(\overline{\mathbf{M}})=\{\lambda_{3}+\lambda_{2},\lambda_{3}+\lambda_{1},\lambda_{2}+\lambda_{1}\}$
such that $\underline{\lambda}_{\overline{\mathbf{M}}}=\lambda_{2}+\lambda_{1}>0$.
\begin{defn}
	\label{def:Unstable-set}Define an unstable set $\mathbb{U}_{s}\subset\mathbb{SO}\left(3\right)$
	as
	\begin{equation}
	\mathbb{U}_{s}=\left\{ \left.\tilde{R}(0)\in\mathbb{SO}\left(3\right)\right|{\rm Tr}\{\tilde{R}(0)\}=-1\right\} \label{eq:Unstable_set_SO3}
	\end{equation}
	where $\tilde{R}(0)\in\mathbb{U}_{s}$ if one of the following conditions
	is met: $\tilde{R}(0)={\rm diag}(1,-1,-1)$, $\tilde{R}(0)={\rm diag}(-1,1,-1)$,
	or $\tilde{R}(0)={\rm diag}(-1,-1,1)$ which indicates that $||\tilde{R}(0)||_{{\rm I}}=+1$.
\end{defn}

\section{Nonlinear Stochastic Navigation Observer\label{sec:Non-Nav-Observer1}}

In view of the vector measurements in \eqref{eq:NAV_Set_Measurements},
and the true compact dynamics defined in \eqref{eq:NAV_Detailed_True_dot},
we propose the following nonlinear stochastic navigation observer
with known gravity developed on the matrix Lie Group of $\mathbb{SE}_{2}\left(3\right)$
and compactly expressed as:
\begin{equation}
\dot{\hat{X}}=\hat{X}U_{m}-W\hat{X}\label{eq:NAV_Filter1}
\end{equation}
with $U_{m}=u([\Omega_{m}\text{\ensuremath{]_{\times}}},0_{3\times1},a_{m},1)=\left[\begin{array}{ccc}
[\Omega_{m}\text{\ensuremath{]_{\times}}} & 0_{3\times1} & a_{m}\\
0_{1\times3} & 0 & 0\\
0_{1\times3} & 1 & 0
\end{array}\right]\in\mathcal{U}_{\mathcal{M}}$, $\hat{X}\in\mathbb{SE}_{2}\left(3\right)$ being the estimate of
$X$, and $W=u([w_{\Omega}\text{\ensuremath{]_{\times}}},w_{V},w_{a},1)=\left[\begin{array}{ccc}
[w_{\Omega}\text{\ensuremath{]_{\times}}} & w_{V} & w_{a}\\
0_{1\times3} & 0 & 0\\
0_{1\times3} & 1 & 0
\end{array}\right]\in\mathcal{U}_{\mathcal{M}}$ being a matrix composed of correction factors. It becomes apparent
that $\dot{\hat{X}}:\mathbb{SE}_{2}\left(3\right)\times\mathcal{U}_{\mathcal{M}}\rightarrow T_{\hat{X}}\mathbb{SE}_{2}\left(3\right)\subset\mathbb{R}^{5\times5}$.
The observer in \eqref{eq:NAV_Filter1} can be detailed as follows:
\begin{equation}
\begin{cases}
\dot{\hat{R}} & =\hat{R}[\Omega_{m}\text{\ensuremath{]_{\times}}}-[w_{\Omega}]_{\times}\hat{R}\\
\dot{\hat{P}} & =\hat{V}-[w_{\Omega}]_{\times}\hat{P}-w_{V}\\
\dot{\hat{V}} & =\hat{R}a_{m}-[w_{\Omega}]_{\times}\hat{V}-w_{a}\\
w_{\Omega} & =-k_{w}(||M\tilde{R}||_{{\rm I}}+1)\boldsymbol{\Upsilon}(M\tilde{R})\\
& \hspace{1em}-\frac{1}{4}\frac{||M\tilde{R}||_{{\rm I}}+2}{||M\tilde{R}||_{{\rm I}}+1}\hat{R}{\rm diag}(\hat{R}^{\top}\boldsymbol{\Upsilon}(M\tilde{R}))\hat{\sigma}\\
w_{V} & =[p_{c}]_{\times}w_{\Omega}-k_{v}\tilde{R}^{\top}\tilde{P}_{\varepsilon}\\
w_{a} & =-\overrightarrow{\mathtt{g}}-k_{a}\tilde{R}^{\top}\tilde{P}_{\varepsilon}\\
k_{R} & =\gamma_{\sigma}\frac{||M\tilde{R}||_{{\rm I}}+2}{8}\exp(||M\tilde{R}||_{{\rm I}})\\
\dot{\hat{\sigma}}_{\Omega} & =k_{R}{\rm diag}(\hat{R}^{\top}\boldsymbol{\Upsilon}(M\tilde{R}))\hat{R}^{\top}\boldsymbol{\Upsilon}(M\tilde{R})-k_{\sigma}\gamma_{\sigma}\hat{\sigma}
\end{cases}\label{eq:NAV_Filter1_Detailed}
\end{equation}
with $k_{w}$, $k_{v}$, $k_{a}$, $\gamma_{\sigma}$, and $k_{\sigma}$
being positive constants. Quaternion representation of \eqref{eq:NAV_Filter1_Detailed}
is presented in the \nameref{subsec:Appendix-A}.
\begin{thm}
	\label{thm:Theorem1}Consider the stochastic system in \eqref{eq:NAV_Detailed_True_STCH_dot}.
	Let Assumption \ref{Assum:NAV_1Landmark} hold true. Let the nonlinear
	navigation stochastic observer in \eqref{eq:NAV_Filter1} be combined
	with the set of measurements in \eqref{eq:NAV_Set_Measurements} along
	with $\overline{y}_{i}=X^{-1}\overline{p}_{i}$, $\Omega_{m}=\Omega+n_{\Omega}$,
	and $a_{m}=a+n_{a}$. Hence, for $\tilde{R}(0)\notin\mathbb{U}_{s}$
	defined in \eqref{eq:Unstable_set_SO3}, all the signals in the closed-loop
	are almost semi-globally uniformly ultimately bounded in the mean
	square.
\end{thm}
\begin{proof}Considering \eqref{eq:NAV_Detailed_True_STCH_dot} and
	\eqref{eq:NAV_Filter1_Detailed}, one obtains
	
	\begin{equation}
	\begin{cases}
	d\tilde{R} & =\tilde{R}[w_{\Omega}]_{\times}dt-\tilde{R}[\hat{R}\mathcal{Q}d\beta_{\Omega}]_{\times}\\
	d\tilde{P} & =(\tilde{V}+\tilde{R}w_{V})dt+\tilde{R}[\hat{P}]_{\times}\hat{R}\mathcal{Q}d\beta_{\Omega}\\
	d\tilde{V} & =((\mathbf{I}_{3}-\tilde{R})g+\tilde{R}w_{a})dt-\tilde{R}\hat{R}\mathcal{Q}d\beta_{a}\\
	& \hspace{1em}-\tilde{R}[\hat{V}]_{\times}\hat{R}\mathcal{Q}d\beta
	\end{cases}\label{eq:NAV_Filter1_Error_dot}
	\end{equation}
	Thus, it is straightforward to show that
	\begin{equation}
	\begin{cases}
	d||M\tilde{R}||_{{\rm I}} & =\underbrace{\frac{1}{2}\boldsymbol{\Upsilon}(M\tilde{R})^{\top}w_{\Omega}}_{f_{R}}dt+\underbrace{-\frac{1}{2}\boldsymbol{\Upsilon}(M\tilde{R})^{\top}\hat{R}}_{h_{R}}\mathcal{Q}d\beta_{\Omega}\\
	d\tilde{R}^{\top}\tilde{P}_{\varepsilon} & =\underbrace{(\tilde{R}^{\top}\tilde{V}-[p_{c}-\tilde{R}^{\top}\tilde{P}_{\varepsilon}]_{\times}w_{\Omega}+w_{V})}_{f_{P}}dt\\
	& \hspace{1em}+\underbrace{-[\hat{P}-p_{c}+\tilde{R}^{\top}\tilde{P}_{\varepsilon}]_{\times}\hat{R}}_{h_{P}}\mathcal{Q}d\beta_{\Omega}\\
	d\tilde{R}^{\top}\tilde{V} & =\underbrace{(-[w_{\Omega}]_{\times}\tilde{R}^{\top}\tilde{V}+(\tilde{R}-\mathbf{I}_{3})^{\top}\overrightarrow{\mathtt{g}}+w_{a})}_{f_{V}}dt\\
	& \hspace{1em}+\underbrace{-[\begin{array}{cc}
		[\tilde{R}^{\top}V]_{\times}\hat{R} & \hat{R}\end{array}]}_{h_{V}}[\begin{array}{cc}
	\mathcal{Q}d\beta_{\Omega} & \mathcal{Q}d\beta_{a}\end{array}]^{\top}
	\end{cases}\label{eq:NAV_Filter1_Error_dotf}
	\end{equation}
	Let $\mathbb{V}=\mathbb{V}(||M\tilde{R}||_{{\rm I}},\tilde{R}^{\top}\tilde{P}_{\varepsilon},\tilde{R}^{\top}\tilde{V},\tilde{\sigma})$
	be a Lyapunov function candidate given by
	\begin{equation}
	\mathbb{V}=\mathbb{V}^{a}+\mathbb{V}^{b}\label{eq:NAV_LyapT_1}
	\end{equation}
	The real-valued function $\mathbb{V}^{a}$ has the map $\mathbb{V}^{a}:\mathbb{SO}\left(3\right)\times\mathbb{R}^{3}\rightarrow\mathbb{R}_{+}$
	defined by
	\begin{equation}
	\mathbb{V}^{a}=\exp(||M\tilde{R}||_{{\rm I}})||M\tilde{R}||_{{\rm I}}+\frac{1}{2\gamma_{\sigma}}||\tilde{\sigma}||^{2}\label{eq:NAV_LyapR_1}
	\end{equation}
	In view of \eqref{eq:NAV_Vfunction_Lyap0}, one easily finds that
	$\mathbb{V}_{||M\tilde{R}||_{{\rm I}}}^{a}=(||M\tilde{R}||_{{\rm I}}+1)\exp(E_{R})$
	and $\mathbb{V}_{||M\tilde{R}||_{{\rm I}}||M\tilde{R}||_{{\rm I}}}^{a}=(||M\tilde{R}||_{{\rm I}}+2)\exp(E_{R})$.
	From \eqref{eq:NAV_dVfunction_Lyap} and \eqref{eq:NAV_Filter1_Error_dotf},
	one {\small{}obtains
		\begin{align}
		\mathcal{L}\mathbb{V}^{a}= & \mathbb{V}_{||M\tilde{R}||_{{\rm I}}}^{a^{\top}}f_{R}+\frac{1}{2}{\rm Tr}\{h_{R}\mathcal{Q}^{2}h_{R}^{\top}\mathbb{V}_{||M\tilde{R}||_{{\rm I}}||M\tilde{R}||_{{\rm I}}}^{a}\}-\frac{1}{\gamma_{\sigma}}\tilde{\sigma}^{\top}\dot{\hat{\sigma}}\nonumber \\
		\leq & \frac{1}{2}\exp(||M\tilde{R}||_{{\rm I}})(||M\tilde{R}||_{{\rm I}}+1)\boldsymbol{\Upsilon}(M\tilde{R})^{\top}w_{\Omega}\nonumber \\
		& +k_{R}\boldsymbol{\Upsilon}(M\tilde{R})^{\top}\hat{R}{\rm diag}(\sigma)\hat{R}^{\top}\boldsymbol{\Upsilon}(M\tilde{R})-\frac{1}{\gamma_{\sigma}}\tilde{\sigma}^{\top}\dot{\hat{\sigma}}\label{eq:NAV_LyapR_1dot_1}
		\end{align}
	}where $\mathcal{Q}^{2}\leq{\rm diag}(\sigma)$ as in \eqref{eq:NAV_s}.
	Replacing $w_{\Omega}$ and $\dot{\hat{\sigma}}$ with their definitions
	in \eqref{eq:NAV_Filter1_Detailed} and considering \eqref{eq:SLAM_lemm1_1}
	in Lemma \ref{Lemm:SLAM_Lemma1}, one find{\small{}s
		\begin{align}
		& \mathcal{L}\mathbb{V}^{a}\leq-(1+{\rm Tr}\{\tilde{R}\})\frac{k_{w}\underline{\lambda}_{\overline{\mathbf{M}}}}{4}\exp(||M\tilde{R}||_{{\rm I}})||M\tilde{R}||_{{\rm I}}+k_{\sigma}\tilde{\sigma}^{\top}\hat{\sigma}\nonumber \\
		& \hspace{1em}\leq-k_{w}c_{R}\exp(||M\tilde{R}||_{{\rm I}})||M\tilde{R}||_{{\rm I}}-\frac{k_{\sigma}}{2}||\tilde{\sigma}||^{2}+\frac{k_{\sigma}}{2}||\sigma||^{2}\label{eq:NAV_LyapR_1dot_2}
		\end{align}
	}where $k_{\sigma}\tilde{\sigma}^{\top}\sigma\leq(k_{\sigma}/2)||\sigma||^{2}+(k_{\sigma}/2)||\tilde{\sigma}||^{2}$
	as to Young's inequality and $c_{R}=\frac{1}{4}\underline{\lambda}_{\overline{\mathbf{M}}}(1+{\rm Tr}\{\tilde{R}\})$.
	Bringing our attention to the second part of \eqref{eq:NAV_LyapT_1},
	the real-valued function $\mathbb{V}^{b}$ has the map $\mathbb{V}^{b}:\mathbb{SO}\left(3\right)\times\mathbb{R}^{3}\times\mathbb{R}^{3}\rightarrow\mathbb{R}_{+}$
	defined by
	\begin{align}
	\mathbb{V}^{b}= & \frac{||\tilde{R}^{\top}\tilde{P}_{\varepsilon}||^{4}}{4}+\frac{||\tilde{R}^{\top}\tilde{V}||^{4}}{4k_{a}}-\frac{||\tilde{R}^{\top}\tilde{V}||^{2}\tilde{V}^{\top}\tilde{R}\tilde{R}^{\top}\tilde{P}_{\varepsilon}}{\mu}\label{eq:NAV_LyapPV_1}
	\end{align}
	where $k_{a}$ and $\mu$ are positive constants. From \eqref{eq:NAV_Vfunction_Lyap0},
	one has
	\begin{align}
	\mathbb{V}_{\tilde{R}^{\top}\tilde{P}_{\varepsilon}}^{b} & =||\tilde{R}^{\top}\tilde{P}_{\varepsilon}||^{2}\tilde{R}^{\top}\tilde{P}_{\varepsilon}-\frac{1}{\mu}||\tilde{R}^{\top}\tilde{V}||^{2}\tilde{R}^{\top}\tilde{V}\nonumber \\
	\mathbb{V}_{\tilde{R}^{\top}\tilde{P}_{\varepsilon}\tilde{R}^{\top}\tilde{P}_{\varepsilon}}^{b} & =||\tilde{R}^{\top}\tilde{P}_{\varepsilon}||^{2}\mathbf{I}_{3}+2\tilde{R}^{\top}\tilde{P}_{\varepsilon}\tilde{P}_{\varepsilon}^{\top}\tilde{R}\nonumber \\
	\mathbb{V}_{\tilde{R}^{\top}\tilde{V}}^{b} & =\frac{1}{k_{d}}||\tilde{R}^{\top}\tilde{V}||^{2}\tilde{R}^{\top}\tilde{V}-\frac{1}{\mu}||\tilde{R}^{\top}\tilde{V}||^{2}\tilde{R}^{\top}\tilde{P}_{\varepsilon}\nonumber \\
	\mathbb{V}_{\tilde{R}^{\top}\tilde{V}\tilde{R}^{\top}\tilde{V}}^{b} & =\frac{1}{k_{d}}||\tilde{R}^{\top}\tilde{V}||^{2}\mathbf{I}_{3}-\frac{2}{\mu}\tilde{R}^{\top}\tilde{P}_{\varepsilon}\tilde{V}^{\top}\tilde{R}\label{eq:NAV_LyapPV_vv}
	\end{align}
	Therefore, from \eqref{eq:NAV_dVfunction_Lyap}, \eqref{eq:NAV_Filter1_Error_dotf},
	\eqref{eq:NAV_LyapPV_1}, and \eqref{eq:NAV_LyapPV_vv}, one obtains
	\begin{align}
	& \mathcal{L}\mathbb{V}^{b}=\mathbb{V}_{\tilde{R}^{\top}\tilde{P}_{\varepsilon}}^{b^{\top}}f_{P}+\frac{1}{2}{\rm Tr}\{h_{P}\mathcal{Q}^{2}h_{P}^{\top}\mathbb{V}_{\tilde{R}^{\top}\tilde{P}_{\varepsilon}\tilde{R}^{\top}\tilde{P}_{\varepsilon}}^{b}\}\nonumber \\
	& \hspace{1em}+\mathbb{V}_{\tilde{R}^{\top}\tilde{V}}^{b^{\top}}f_{V}+\frac{1}{2}{\rm Tr}\{h_{V}\mathcal{Q}^{2}h_{V}^{\top}\mathbb{V}_{\tilde{R}^{\top}\tilde{V}\tilde{R}^{\top}\tilde{V}}^{b}\}\nonumber \\
	& \leq-(k_{v}-c_{1})||\tilde{R}^{\top}\tilde{P}_{\varepsilon}||^{4}-(1/\mu-c_{2})||\tilde{R}^{\top}\tilde{V}||^{4}\nonumber \\
	& \hspace{1em}+(k_{d}k_{v}/\mu^{2}+c_{3})||\tilde{R}^{\top}\tilde{V}||^{2}||\tilde{R}^{\top}\tilde{P}_{\varepsilon}||^{2}\nonumber \\
	& \hspace{1em}+c_{g}||\tilde{R}^{\top}\tilde{V}||^{2}||\mathbf{I}_{3}-\tilde{R}||_{F}+(\frac{3c_{P}}{2}+\frac{c_{V}}{4k_{d}})||\sigma||^{2}\label{eq:NAV_LyapPV_1dot_1}
	\end{align}
	where $c_{1}=\frac{3c_{P}}{4}+\frac{1}{2}$, $c_{2}=\frac{||\overrightarrow{\mathtt{g}}||}{2k_{d}}+\frac{1}{4k_{d}}$,
	$c_{3}=\frac{||\overrightarrow{\mathtt{g}}||}{2\mu}+\frac{1}{4k_{d}}$,
	$c_{g}=\frac{||\overrightarrow{\mathtt{g}}||}{2k_{d}}+\frac{||\overrightarrow{\mathtt{g}}||}{2\mu}$,
	$c_{V}=\sup_{t\geq0}(1+||V||^{2})$, and $c_{P}=\sup_{t\geq0}||P-p_{c}||^{2}$.
	Also, one finds ${\rm Tr}\{\hat{R}\mathcal{Q}^{2}\hat{R}^{\top}\}={\rm Tr}\{\mathcal{Q}^{2}\}$,
	$\tilde{V}^{\top}\tilde{R}\left[w_{\Omega}\right]_{\times}\tilde{R}^{\top}\tilde{V}=0$,
	and $||\mathbf{I}_{3}-\tilde{R}||_{F}=2\sqrt{2}\sqrt{||\tilde{R}||_{{\rm I}}}\leq4\overline{\lambda}_{M}\sqrt{||M\tilde{R}||_{{\rm I}}}$.
	Hence, the inequality in \eqref{eq:NAV_LyapPV_1dot_1} becomes
	\begin{align}
	& \mathcal{L}\mathbb{V}^{b}\leq-e_{1}^{\top}\underbrace{\left[\begin{array}{cc}
		k_{v}-c_{1} & \frac{1}{2}(k_{d}k_{v}/\mu^{2}+c_{3})\\
		\frac{1}{2}(k_{d}k_{v}/\mu^{2}+c_{3}) & \frac{1}{\mu}-c_{2}
		\end{array}\right]}_{A_{1}}e_{1}\nonumber \\
	& \hspace{1.7em}+c_{g}||\tilde{R}^{\top}\tilde{V}||^{2}||\mathbf{I}_{3}-\tilde{R}||_{F}+(\frac{3c_{P}}{2}+\frac{c_{V}}{4k_{d}})||\sigma||^{2}\label{eq:NAV_LyapPV_1dot_Final}
	\end{align}
	where $e_{1}=[||\tilde{R}^{\top}\tilde{P}_{\varepsilon}||^{2},||\tilde{R}^{\top}\tilde{V}||^{2}]^{\top}$.
	It can be deduced that $A_{1}$ can be made positive by selecting
	$k_{v}>3/4$, $k_{v}>\mu c_{x}/\underline{c}_{1}$, and $k_{v}>c_{x}\mu^{2}/(\mu\underline{c}_{1}-c_{x})$.
	Considering the parameter selection above, define $\underline{\lambda}_{1}=\underline{\lambda}(A_{1})$.
	In view of \eqref{eq:NAV_LyapT_1}, \eqref{eq:NAV_LyapR_1dot_2},
	and \eqref{eq:NAV_LyapPV_1dot_Final}, the differential operator $\mathcal{L}\mathbb{V}$
	can be expressed as
	\begin{align}
	\mathcal{L}\mathbb{V}\leq & -k_{w}c_{R}\exp(||M\tilde{R}||_{{\rm I}})||M\tilde{R}||_{{\rm I}}-(k_{\sigma}/2)||\tilde{\sigma}||^{2}\nonumber \\
	& +(k_{\sigma}/2)||\sigma||^{2}-\underline{\lambda}_{1}||\tilde{R}^{\top}\tilde{P}_{\varepsilon}||^{4}-\underline{\lambda}_{1}||\tilde{R}^{\top}\tilde{V}||^{4}\nonumber \\
	& +c_{g}||\tilde{R}^{\top}\tilde{V}||^{2}||\mathbf{I}_{3}-\tilde{R}||_{F}+(\frac{3c_{P}}{2}+\frac{c_{V}}{4k_{d}})||\sigma||^{2}\nonumber \\
	\leq & -e_{2}^{\top}\underbrace{\left[\begin{array}{c|c}
		k_{w}c_{R} & \frac{c_{g}}{2}\mathbf{1}_{1\times2}\\
		\hline \frac{c_{g}}{2}\mathbf{1}_{2\times1} & \underline{\lambda}_{1}\mathbf{I}_{2}
		\end{array}\right]}_{A_{2}}e_{2}-(k_{\sigma}/2)||\tilde{\sigma}||^{2}\nonumber \\
	& +\eta_{2}||\sigma||^{2}\label{eq:NAV_LyapPV_1dot_2}
	\end{align}
	where $e_{2}=[\sqrt{\exp(||M\tilde{R}||_{{\rm I}})||M\tilde{R}||_{{\rm I}}},||\tilde{R}^{\top}\tilde{P}_{\varepsilon}||^{2},||\tilde{R}^{\top}\tilde{V}||^{2}]^{\top}$
	and $\eta_{2}=\frac{3c_{P}}{2}+\frac{c_{V}}{4k_{d}}$. To make $A_{2}$
	positive, consider selecting $k_{w}c_{R}\underline{\lambda}_{1}>c_{g}^{2}/4$.
	Define {\footnotesize{}$e_{T}=[\sqrt{\exp(||M\tilde{R}||_{{\rm I}})||M\tilde{R}||_{{\rm I}}},||\tilde{R}^{\top}\tilde{P}_{\varepsilon}||^{2},||\tilde{R}^{\top}\tilde{V}||^{2},||\tilde{\sigma}||]^{\top}$}
	and $\eta_{1}=\min\{\underline{\lambda}(A_{2}),k_{\sigma}/2\}$. Hence,
	one has
	\begin{align}
	\mathcal{L}\mathbb{V}\leq & -\eta_{1}||e_{T}||^{2}+\eta_{2}||\sigma||^{2}\label{eq:NAV_LyapT_1dot_Final}
	\end{align}
	such that
	\begin{equation}
	d\mathbb{E}[\mathbb{V}]/dt=\mathbb{E}[\mathcal{L}\mathbb{V}]\leq-\eta_{1}\mathbb{E}[\mathbb{V}]+\eta_{2}\label{eq:SLAM_Lyap2_final3}
	\end{equation}
	In accordance with Lemma \ref{Lemm:NAV_deng}, it becomes apparent
	that
	\[
	0\leq\mathbb{E}[\mathbb{V}(t)]\leq\mathbb{V}(0){\rm exp}(-\eta_{1}t)+\text{\ensuremath{\eta_{2}}}/\eta_{1},\,\forall t\geq0
	\]
	Thus, it can be seen that the vector $e_{T}$ is almost SGUUB which
	completes the proof.\end{proof}

\subsection{Nonlinear Stochastic Observer with Unknown Gravity}

Consider an unknown gravity vector $\overrightarrow{\mathtt{g}}$
and let $\hat{\mathtt{g}}$ denote the estimate of $\overrightarrow{\mathtt{g}}$.
Define the error between $\hat{\mathtt{g}}$ and $\overrightarrow{\mathtt{g}}$
as $\tilde{\mathtt{g}}=\overrightarrow{\mathtt{g}}-\tilde{R}\hat{\mathtt{g}}$.
Modify $w_{a}$ in the observer design in \eqref{eq:NAV_Filter1_Detailed}
to include $\dot{\hat{\mathtt{g}}}$ as follows:
\begin{align}
w_{a} & =-\hat{\mathtt{g}}-k_{a}\tilde{R}^{\top}\tilde{P}_{\varepsilon}\nonumber \\
\dot{\hat{\mathtt{g}}} & =-[w_{\Omega}]_{\times}\hat{\mathtt{g}}+\mu\gamma_{g}\tilde{R}^{\top}\tilde{P}_{\varepsilon}\label{eq:SLAM_Gravity}
\end{align}
where $\gamma_{g}>0$ is an adaptation gain. Let $\mathbb{V}=\mathbb{V}(||M\tilde{R}||_{{\rm I}},\tilde{R}^{\top}\tilde{P}_{\varepsilon},\tilde{R}^{\top}\tilde{V},\tilde{R}^{\top}\tilde{\mathtt{g}},\tilde{\sigma})$
be a Lyapunov function candidate given by $\mathbb{V}=\mathbb{V}^{a}+\mathbb{V}^{b}$
where $\mathbb{V}^{a}$ is as in \eqref{eq:NAV_LyapR_1} while $\mathbb{V}^{b}=\frac{||\tilde{R}^{\top}\tilde{P}_{\varepsilon}||^{4}}{4}+\frac{||\tilde{R}^{\top}\tilde{V}||^{4}}{4k_{a}}+\frac{||\tilde{R}^{\top}\tilde{\mathtt{g}}||^{2}}{2\gamma_{g}}-\frac{||\tilde{R}^{\top}\tilde{V}||^{2}\tilde{V}^{\top}\tilde{R}\tilde{R}^{\top}\tilde{P}_{\varepsilon}}{\mu}-\frac{||\tilde{R}^{\top}\tilde{V}||^{2}\tilde{V}^{\top}\tilde{R}\tilde{R}^{\top}\tilde{\mathtt{g}}}{\mu}$.
Following the analogous proving steps of Theorem \ref{thm:Theorem1}
the obtained result is similar to \eqref{eq:SLAM_Lyap2_final3}.

The detailed implementation steps of the observer in its discrete
form can be found in Algorithm \ref{alg:Alg_Disc0}, where $\Delta t$
denotes a small sample time.
\begin{algorithm}
	\caption{\label{alg:Alg_Disc0}Discrete nonlinear stochastic observer}
	
	\textbf{Initialization}:
	\begin{enumerate}
		\item[{\footnotesize{}1:}] Set $\hat{R}_{0|0}\in\mathbb{SO}\left(3\right)$, and $\hat{P}_{0|0},\hat{V}_{0|0},\hat{\sigma}_{0|0},\hat{\mathtt{g}}_{0}\in\mathbb{R}^{3}$
		\item[{\footnotesize{}2:}] Start with $k=0$ and select the design parameters
	\end{enumerate}
	\textbf{while }(1)\textbf{ do}
	\begin{enumerate}
		\item[{\footnotesize{}3:}] $\hat{X}_{k|k}=\left[\begin{array}{ccc}
		\hat{R}_{k|k} & \hat{P}_{k|k} & \hat{V}_{k|k}\\
		0_{1\times3} & 1 & 0\\
		0_{1\times3} & 0 & 1
		\end{array}\right]$ and \\
		$\hat{U}_{k}=\left[\begin{array}{ccc}
		[\Omega_{m}[k]\text{\ensuremath{]_{\times}}} & 0_{3\times1} & a_{m}[k]\\
		0_{1\times3} & 0 & 0\\
		0_{1\times3} & 1 & 0
		\end{array}\right]$
		\item[] \textcolor{blue}{/{*} Prediction {*}/}
		\item[{\footnotesize{}4:}] $\hat{X}_{k+1|k}=\hat{X}_{k|k}\exp(\hat{U}_{k}\Delta t)$
		\item[] \textcolor{blue}{/{*} Update step {*}/}
		\item[{\footnotesize{}5:}] {\small{}$\begin{cases}
			p_{c} & =\frac{1}{s_{T}}\sum_{i=1}^{n}s_{i}p_{i}[k],\hspace{1em}s_{T}=\sum_{i=1}^{n}s_{i}\\
			M_{k} & =\sum_{i=1}^{n}s_{i}p_{i}[k]p_{i}^{\top}[k]-s_{T}p_{c}p_{c}^{\top}\\
			M\tilde{R}_{k} & =\sum_{i=1}^{n}s_{i}\left(p_{i}[k]-p_{c}\right)y_{i}^{\top}[k]\hat{R}_{k+1|k}^{\top}\\
			\tilde{R}^{\top}\tilde{P}_{\varepsilon}[k] & =\frac{1}{s_{T}}\sum_{i=1}^{n}s_{i}(p_{i}[k]-\hat{R}_{k+1|k}y_{i}[k]-\hat{P}_{k+1|k})
			\end{cases}$}{\small\par}
		\item[{\footnotesize{}6:}] $\begin{cases}
		w_{\Omega}[k] & =-k_{w}(||M\tilde{R}_{k}||_{{\rm I}}+1)\boldsymbol{\Upsilon}(M\tilde{R}_{k})\\
		& \hspace{1em}-\frac{1}{4}\frac{||M\tilde{R}_{k}||_{{\rm I}}+2}{||M\tilde{R}_{k}||_{{\rm I}}+1}\hat{R}_{k}{\rm diag}(\hat{R}_{k}^{\top}\boldsymbol{\Upsilon}(M\tilde{R}_{k}))\hat{\sigma}_{k}\\
		w_{V}[k] & =[p_{c}[k]]_{\times}w_{\Omega}[k]-k_{v}\tilde{R}^{\top}\tilde{P}_{\varepsilon}[k]\\
		& {\color{blue}\text{ /* If gravity is known, \ensuremath{\hat{\mathtt{g}}_{k+1}=\ensuremath{\overrightarrow{\mathtt{g}}}} */}}\\
		\hat{\mathtt{g}}_{k+1} & =\hat{\mathtt{g}}_{k}\Delta t(-[w_{\Omega}[k]]_{\times}\hat{\mathtt{g}}_{k}+\mu\gamma_{g}\tilde{R}^{\top}\tilde{P}_{\varepsilon}[k])\\
		w_{a}[k] & =-\hat{\mathtt{g}}_{k+1}-k_{a}\tilde{R}^{\top}\tilde{P}_{\varepsilon}[k]
		\end{cases}$
		\item[{\footnotesize{}7:}] $W_{k}=\left[\begin{array}{ccc}
		[w_{\Omega}[k]\text{\ensuremath{]_{\times}}} & w_{V}[k] & w_{a}[k]\\
		0_{1\times3} & 0 & 0\\
		0_{1\times3} & 1 & 0
		\end{array}\right]$
		\item[{\footnotesize{}8:}] $k_{R}=\gamma_{\sigma}\frac{||M\tilde{R}_{k}||_{{\rm I}}+2}{8}\exp(||M\tilde{R}_{k}||_{{\rm I}})$
		\item[{\footnotesize{}9:}] {\small{}$\hat{\sigma}_{k+1}=\hat{\sigma}_{k}+\Delta tk_{R}{\rm diag}(\hat{R}_{k+1|k}^{\top}\boldsymbol{\Upsilon}(M\tilde{R}_{k}))\hat{R}_{k+1|k}^{\top}\boldsymbol{\Upsilon}(M\tilde{R}_{k})$}{\small\par}
		\item[] $\hspace{1em}\hspace{1em}\hspace{1em}-\Delta tk_{\sigma}\gamma_{\sigma}\hat{\sigma}_{k}$
		\item[{\footnotesize{}10:}] $\hat{X}_{k+1|k+1}=\exp(-W_{k}\Delta t)\hat{X}_{k+1|k}$
		\item[{\footnotesize{}11:}] $k=k+1$
	\end{enumerate}
	\textbf{end while}
\end{algorithm}

\section{Experimental Results \label{sec:SE3_Simulations}}

This section experimentally evaluates the performance of the proposed
nonlinear stochastic navigation observers on the Lie group of $\mathbb{SE}_{2}\left(3\right)$.
The discrete forms of the proposed observers (known gravity and unknown
gravity) outlined in Algorithm \ref{alg:Alg_Disc0} have been examined
using real-word data (EuRoc dataset) \cite{Burri2016Euroc}. The dataset
contains the ground truth, IMU measurements obtained by ADIS16448
at a sampling rate of 200 Hz, and stereo images obtained by MT9V034
at a sampling rate of 20 Hz. Owing to the fact that landmark positions
are not included in the dataset, landmarks are placed arbitrarily.
To increase the rigor of the experiment, IMU measurements were supplemented
with additional normally distributed noise $n_{\Omega}=\mathcal{N}(0,0.12)$
(rad/sec) and $n_{a}=\mathcal{N}(0,0.11)$ (m/sec$^{2}$) with a zero
mean and a standard deviation $0.12$ and $0.11$, respectively. The
design parameters are selected as follows: $k_{w}=3$, $k_{v}=10$,
$k_{a}=10$, $\gamma_{\sigma}=3$, $\gamma_{\mathtt{g}}=2$, and $k_{\sigma}=0.1$,
while the initial covariance estimate is $\hat{\sigma}\left(0\right)=\hat{\mathtt{g}}\left(0\right)=[0,0,0]^{\top}$.

Fig. \ref{fig:NAV_3D} shows strong tracking capabilities in view
of uncertain measurements and large initialization error in attitude
and position. Fig. \ref{fig:NAV_ERR} demonstrates fast convergence
of the error components $||R\hat{R}^{\top}||_{{\rm I}}$, $||P-\hat{P}||$,
$||V-\hat{V}||$, and $||\mathtt{g}-\hat{\mathtt{g}}||$ from large
values to the close neighborhood of the origin. It can be noticed
that impressive results have been achieved at low sampling rates demonstrating
the computational inexpensiveness of the proposed algorithm.

\begin{figure}[h]
	\centering{}\includegraphics[scale=0.33]{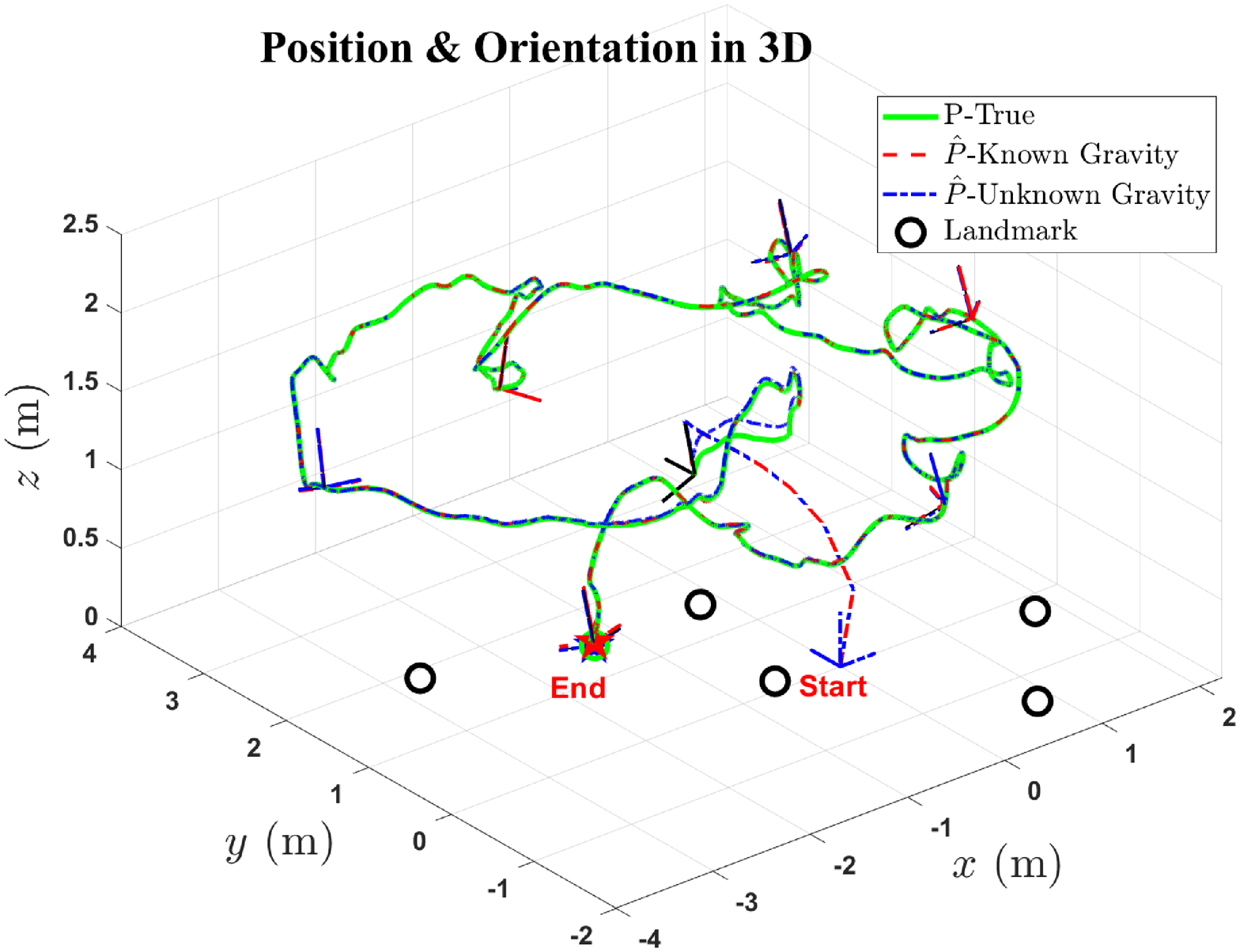}\caption{Vicon Room 2 01 experimental validation using a dataset. The true
		trajectories (green solid-line) and three axes true attitude (black
		solid-line) is plotted against the trajectory estimated by the proposed
		nonlinear stochastic discrete navigation observers (Algorithm \ref{alg:Alg_Disc0};
		red dashed-line and blue center-line). The landmarks are plotted as
		black circles.}
	\label{fig:NAV_3D}
\end{figure}

\begin{figure}[h]
	\centering{}\includegraphics[scale=0.3]{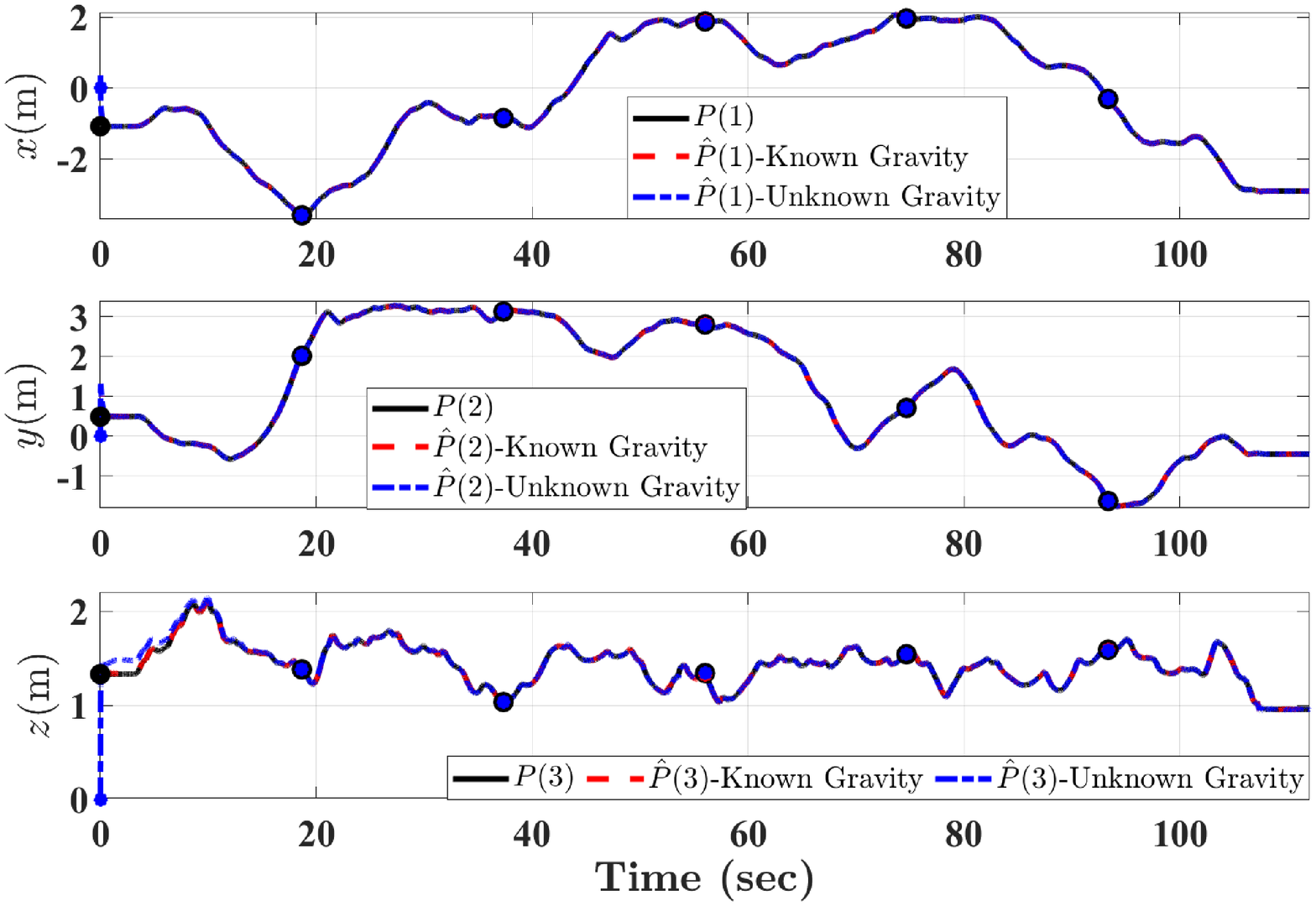}\caption{Position: true (black solid-line) vs estimated observers (red dashed-line
		and blue center-line).}
	\label{fig:NAV_Pos}
\end{figure}

\begin{figure}[h]
	\centering{}\includegraphics[scale=0.28]{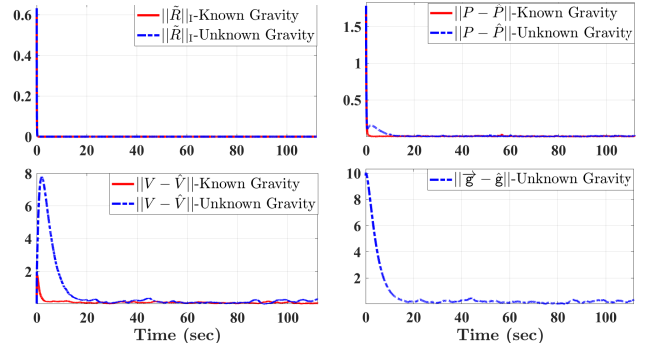}\caption{Error of proposed observers: Known Gravity vs Unknown Gravity.}
	\label{fig:NAV_ERR}
\end{figure}

\section{Conclusion \label{sec:SE3_Conclusion}}

This paper addresses the problem of attitude, position, linear velocity,
and gravity estimation of a vehicle traveling with 6 DoF. Nonlinear
stochastic navigation observers on $\mathbb{SE}_{2}(3)$ has been
proposed. The proposed observers are guaranteed to be almost SGUUB
in the mean square. Experimental results revealed robustness and fast
adaptability of the proposed approach for identification of unknown
pose, linear velocity and gravity. 

\section*{Acknowledgment}

The author would like to thank \textbf{Maria Shaposhnikova} for proofreading
the article.

\subsection*{Appendix\label{subsec:Appendix-A}}
\begin{center}
	\textbf{Quaternion of the Proposed Observers}
	\par\end{center}

\noindent Let $Q=[q_{0},q^{\top}]^{\top}\in\mathbb{S}^{3}$ denote
a unit-quaternion vector where $q_{0}\in\mathbb{R}$ and $q\in\mathbb{R}^{3}$
such that $\mathbb{S}^{3}=\{\left.Q\in\mathbb{R}^{4}\right|||Q||=\sqrt{q_{0}^{2}+q^{\top}q}=1\}$.
The inverse of $Q$ is $Q^{-1}=[\begin{array}{cc}
q_{0} & -q^{\top}\end{array}]^{\top}\in\mathbb{S}^{3}$. Let $\odot$ be a quaternion product such that the quaternion multiplication
of $Q_{1}=[\begin{array}{cc}
q_{01} & q_{1}^{\top}\end{array}]^{\top}\in\mathbb{S}^{3}$ and $Q_{2}=[\begin{array}{cc}
q_{02} & q_{2}^{\top}\end{array}]^{\top}\in\mathbb{S}^{3}$ is $Q_{1}\odot Q_{2}=[q_{01}q_{02}-q_{1}^{\top}q_{2},q_{01}q_{2}+q_{02}q_{1}+[q_{1}]_{\times}q_{2}]^{\top}$.
The mapping from unit-quaternion to $\mathbb{SO}\left(3\right)$ is
$\mathcal{R}_{Q}:\mathbb{S}^{3}\rightarrow\mathbb{SO}\left(3\right)$
\begin{align}
\mathcal{R}_{Q} & =(q_{0}^{2}-||q||^{2})\mathbf{I}_{3}+2qq^{\top}+2q_{0}\left[q\right]_{\times}\in\mathbb{SO}\left(3\right)\label{eq:NAV_Append_SO3}
\end{align}
The quaternion identity is $Q_{{\rm I}}=[\pm1,0,0,0]^{\top}$ where
$\mathcal{R}_{Q_{{\rm I}}}=\mathbf{I}_{3}$, see \eqref{eq:NAV_Append_SO3}.
For more details, visit \cite{hashim2019AtiitudeSurvey}. $\hat{Q}=[\hat{q}_{0},\hat{q}^{\top}]^{\top}\in\mathbb{S}^{3}$
is the estimate of $Q=[q_{0},q^{\top}]^{\top}\in\mathbb{S}^{3}$ such
that $\mathcal{R}_{\hat{Q}}=(\hat{q}_{0}^{2}-||\hat{q}||^{2})\mathbf{I}_{3}+2\hat{q}\hat{q}^{\top}+2\hat{q}_{0}\left[\hat{q}\right]_{\times}\in\mathbb{SO}\left(3\right)$.
Considering gravity estimation, the equivalent quaternion representation
of the observer in \eqref{eq:NAV_Filter1_Detailed} and \eqref{eq:SLAM_Gravity}
is as follows:
\[
\begin{cases}
\tilde{y}_{i} & =p_{i}-\mathcal{R}_{\hat{Q}}y_{i}-\hat{P}\\
\Phi_{q} & =M\tilde{R}=\sum_{i=1}^{n}s_{i}\left(p_{i}-p_{c}\right)y_{i}^{\top}\mathcal{R}_{\hat{Q}}^{\top}\\
\boldsymbol{\Upsilon}(\Phi_{q}) & =\mathbf{vex}(\boldsymbol{\mathcal{P}}_{a}(\Phi_{q}))\\
{\rm v}_{q} & =\tilde{R}^{\top}\tilde{P}_{\varepsilon}=\frac{1}{s_{T}}\sum_{i=1}^{n}s_{i}\tilde{y}_{i}\\
||M\tilde{R}||_{{\rm I}} & =\frac{1}{4}{\rm Tr}\{M-\Phi_{q}\}\\
\Theta_{m}= & \left[\begin{array}{cc}
0 & -\Omega_{m}^{\top}\\
\Omega_{m} & -[\Omega_{m}]_{\times}
\end{array}\right],\hspace{1em}\Psi=\left[\begin{array}{cc}
0 & -w_{\Omega}^{\top}\\
w_{\Omega} & [w_{\Omega}]_{\times}
\end{array}\right]\\
\dot{\hat{Q}} & =\frac{1}{2}\Theta_{m}\hat{Q}-\frac{1}{2}\Psi\hat{Q}\\
\dot{\hat{P}} & =\hat{V}-[w_{\Omega}]_{\times}\hat{P}-w_{V}\\
\dot{\hat{V}} & =\mathcal{R}_{\hat{Q}}a_{m}-[w_{\Omega}]_{\times}\hat{V}-w_{a}\\
w_{\Omega} & =-k_{w}(||M\tilde{R}||_{{\rm I}}+1)\boldsymbol{\Upsilon}(\Phi_{q})\\
& -\frac{1}{4}\frac{||M\tilde{R}||_{{\rm I}}+2}{||M\tilde{R}||_{{\rm I}}+1}\mathcal{R}_{\hat{Q}}{\rm diag}(\mathcal{R}_{\hat{Q}}^{\top}\boldsymbol{\Upsilon}(\Phi_{q}))\hat{\sigma}\\
w_{V} & =[p_{c}]_{\times}w_{\Omega}-k_{v}{\rm v}_{q}\\
w_{a} & =-\hat{\mathtt{g}}-k_{a}{\rm v}_{q}\\
\dot{\hat{\mathtt{g}}} & =-[w_{\Omega}]_{\times}\hat{\mathtt{g}}+\mu\gamma_{g}{\rm v}_{q}\\
k_{R} & =(\gamma_{\sigma}/8)(||M\tilde{R}||_{{\rm I}}+2)\exp(||M\tilde{R}||_{{\rm I}})\\
\dot{\hat{\sigma}}_{\Omega} & =k_{R}{\rm diag}(\mathcal{R}_{\hat{Q}}^{\top}\boldsymbol{\Upsilon}(\Phi_{q}))\mathcal{R}_{\hat{Q}}^{\top}\boldsymbol{\Upsilon}(\Phi_{q})-k_{\sigma}\gamma_{\sigma}\hat{\sigma}
\end{cases}
\]

\bibliographystyle{IEEEtran}
\bibliography{bib_Navigation}
% name your BibTeX data base

\end{document}